\documentclass[12pt]{article}

\usepackage{amsthm}
\usepackage{amssymb}
\usepackage{amsmath}
\usepackage{color}
\usepackage[ruled, linesnumbered, noend]{algorithm2e}
\usepackage[hidelinks]{hyperref}

\newcommand{\MS}{\Delta}

\newcommand{\CC}{\mathcal{C}}

\theoremstyle{definition}
\newtheorem{theorem}{Theorem}

\newtheorem{lemma}{Lemma}
\newtheorem{proposition}{Proposition}
\newtheorem{definition}{Definition}

\newtheorem{problem}{Problem}

\title{Listing Small Minimal Separators of a Graph}
\author{Tuukka Korhonen \\ {\small \texttt{tuukka.m.korhonen@helsinki.fi}} \\ {\small \url{https://tuukkakorhonen.com}}}

\begin{document}

\maketitle

\begin{abstract}
Let $G$ be a graph and $a,b$ vertices of $G$.
A minimal $a,b$-separator of $G$ is an inclusion-wise minimal vertex set of $G$ that separates $a$ and $b$.
We consider the problem of enumerating the minimal $a,b$-separators of $G$ that contain at most $k$ vertices, given some integer $k$.
We give an algorithm which enumerates such minimal separators, outputting the first $R$ minimal separators in at most $poly(n) R \cdot \min(4^k, R)$ time for all $R$.
Therefore, our algorithm can be classified as fixed-parameter-delay and incremental-polynomial time.
To the best of our knowledge, no algorithms with non-trivial time complexity have been published for this problem before.
We also discuss barriers for obtaining a polynomial-delay algorithm.
\end{abstract}

\section{Introduction}
Recent state-of-the-art algorithm implementations for determining the tree\-width~\cite{DBLP:conf/sea2/Tamaki19} and the treedepth~\cite{brokkelkamp_et_al:LIPIcs:2020:13332,korhonen:LIPIcs:2020:13333,xu_et_al:LIPIcs:2020:13334} of a graph employ a subroutine that enumerates the minimal separators of the graph that contain at most $k$ vertices, for some bound $k$.
This enumeration is in fact reported as the bottleneck of these implementations.

The problem of enumerating size bounded minimal separators is also a natural refinement of two well-known enumeration problems: the enumeration of (not necessarily minimal) $a,b$-separators of a graph with size at most $k$ and the enumeration of minimal $a,b$-separators without size bound.
Both of them have received significant attention~\cite{DBLP:journals/ijfcs/BerryBC00,DBLP:journals/orl/Hamacher82,DBLP:conf/soda/Kanevsky90,DBLP:journals/siamcomp/KloksK98,DBLP:journals/tcs/ShenL97,DBLP:journals/dam/Takata10} and admit polynomial-delay algorithms~\cite{DBLP:journals/orl/Hamacher82,DBLP:journals/dam/Takata10}.

In this paper we give the following enumeration algorithm.

\begin{theorem}
\label{the:alg}
There is an algorithm that given a graph $G$, a pair of vertices $a,b \in V(G)$, and an integer $k$, enumerates the minimal $a,b$-separators of $G$ that contain at most $k$ vertices, outputting the first $R$ minimal separators in $O^*(R \cdot \min(4^k, R))$\footnote{The $O^*(\cdot)$ notation suppresses factors polynomial in the input size.} time for all $R$.
\end{theorem}

Our technique for obtaining this algorithm is to combine the algorithm of Takata~\cite{DBLP:journals/dam/Takata10} for enumerating minimal separators with the important separators technique developed in~\cite{DBLP:journals/algorithmica/ChenLL09,DBLP:journals/tcs/Marx06} and exposed in~\cite{DBLP:books/sp/CyganFKLMPPS15}.
We obtain our algorithm by using important separators to solve the following decision problem.

\begin{problem}
\label{pro:dec}
Given a graph $G$, a pair of vertices $a,b \in V(G)$, integer $k$, and vertex sets $C \subseteq V(G)$ and $X \subseteq N(C)$ such that $a \in C$ and $G[C]$ is connected, decide if there is a minimal $a,b$-separator $S \subseteq V(G) \setminus C$ with $|S| \le k$ and $X \subseteq S$.
\end{problem}

In particular, Problem~\ref{pro:dec} is the problem of determining if a given subtree of the search tree of Takata's algorithm contains a minimal separator of size at most $k$.
By the $O(n)$ depth of the search tree and standard techniques in enumeration algorithms, solving Problem~\ref{pro:dec} in time $f(G, k, R)$ implies an $f(G, k, R)$-delay algorithm for enumerating minimal $a,b$-separators of size at most $k$.

The following theorem gives evidence why this approach cannot be directly applied to obtain a polynomial-delay algorithm.

\begin{theorem}
\label{the:np}
Problem~\ref{pro:dec} is NP-complete, even when the graph $G$ is bipartite with a bipartition $\{\{a\} \cup N(b), \{b\} \cup N(a)\}$.
\end{theorem}

In this bipartite case of Theorem~\ref{the:np}, size bounded minimal $a,b$-separator enumeration corresponds to size bounded minimal vertex cover enumeration in the graph $G \setminus \{a,b\}$, which to the best of our knowledge is a similarly open problem.
We note that a recent paper gives an \emph{approximate} enumeration algorithm for size bounded minimal vertex cover enumeration~\cite{DBLP:journals/corr/abs-2009-08830}.

\section{Notation}
We use standard graph notation.
A graph $G$ has a vertex set $V(G)$ and edge set $E(G)$.
The subgraph $G[X]$ induced by $X \subseteq V(G)$ has $V(G[X]) = X$ and $E(G[X]) = \{\{u, v\} \in E(G) \mid u,v \in X\}$.
We denote induced subgraphs also by notation $G \setminus X = G[V(G) \setminus X]$.
We denote the vertex sets of connected components of $G$ by $\CC(G)$.
The set of neighbors of a vertex $v$ is denoted by $N(v)$ and the neighborhood of a vertex set $X$ is $N(X) = \bigcup_{v \in X} N(v) \setminus X$.
The set of closed neighbors of $v$ is $N[v] = N(v) \cup \{v\}$ and the closed neighborhood of $X$ is $N[X] = N(X) \cup X$.

For a pair of vertices $a,b \in V(G)$, a minimal $a,b$-separator of $G$ is a vertex set $S \subseteq V(G)$ such that $a$ and $b$ are in different connected components of $G \setminus S$ and $S$ is inclusion-wise minimal with respect to this.
A full component of a set $X \subseteq V(G)$ is a component $C \in \CC(G \setminus X)$ with $N(C) = X$.
It is well-known that $S$ is a minimal $a,b$-separator if and only if $S$ has distinct full components containing $a$ and $b$.

An enumeration algorithm with input $I$ has delay $f(I, R)$ if for all $R$ it outputs the first $R$ solutions in at most $R \cdot f(I, R)$ time.
A polynomial-delay enumeration algorithm has delay $f(I, R) = poly(|I|)$ for some polynomial $poly(|I|)$.
An incremental-polynomial enumeration algorithm has delay $f(I, R) = poly(|I| + R)$ for some polynomial $poly(|I| + R)$.

\section{The Algorithm}
We first discuss Takata's algorithm, then important separators, and then show how these can be combined to obtain our algorithm.
In this section we always consider minimal $a,b$-separators of a graph $G$, so we will not spell this out in our definitions.

\subsection{Takata's Recurrence}
We overview the Takata's recurrence for enumerating all minimal $a,b$-separators with polynomial delay~\cite{DBLP:journals/dam/Takata10}.
We give short proofs for completeness and because our presentation is different from~\cite{DBLP:journals/dam/Takata10}.

\begin{definition}
\label{def:takata}
Let $C$ and $X$ be vertex sets $C \subseteq V(G)$ and $X \subseteq N(C)$ so that $a \in C$ and $G[C]$ is connected.
We denote by $\MS(G, C, X)$ the set of minimal $a,b$-separators $S$ of $G$ such that $S \subseteq V(G) \setminus C$ and $X \subseteq S$.
\end{definition}

The root of Takata's recurrence is $\MS(G) = \MS(G, \{a\}, \emptyset)$.
The leaves of the recurrence have $X = N(C)$, in which case $\MS(G, C, N(C)) = \{N(C)\}$ if $N(C)$ is a minimal $a,b$-separator and $\emptyset$ otherwise.
The internal nodes are defined by the following proposition.

\begin{proposition}[\cite{DBLP:journals/dam/Takata10}]
\label{pro:takata_rec}
Let $C$ and $X$ be as in Definition~\ref{def:takata} and $v$ any vertex in $N(C) \setminus X$.
The sets $\MS(G, C \cup \{v\}, X)$ and $\MS(G, C, X \cup \{v\})$ are disjoint, and their union is equal to $\MS(G, C, X)$.
\end{proposition}
\begin{proof}
The first case corresponds to the minimal separators that do not contain $v$ and the second case to the minimal separators that contain $v$.
\end{proof}

Takata's algorithm uses the search tree constructed by Proposition~\ref{pro:takata_rec}.
To guarantee that the search in this tree finds minimal separators with polynomial delay, it is sufficient to observe that its depth is at most $n$, and to design a polynomial time algorithm for determining if the currently explored subtree is empty, i.e., if $\MS(G, C, X) = \emptyset$.
The following proposition provides this emptiness check.

\begin{proposition}[\cite{DBLP:journals/dam/Takata10}]
\label{pro:takata_empty}
The set $\MS(G, C, X)$ is not empty if and only if there is a component $C_b \in \CC(G \setminus N[C])$ such that $b \in C_b$ and $X \subseteq N(C_b)$.
\end{proposition}
\begin{proof}
If $b \in N[C]$ there is no minimal $a,b$-separator that does not contain vertices in $C$.
Otherwise $N(C_b)$ is a minimal $a,b$-separator in $\MS(G, C, \emptyset)$ because it has a full component $C_b$ containing $b$ and $C_a \supseteq C$ containing $a$ because $N(C_b) \subseteq N(C)$.
Now it suffices to show that if $N(C_b)$ does not subsume $X$ then no minimal separator in $\MS(G, C, \emptyset)$ subsumes $X$.
This follows from the fact that for any such minimal separator $S'$ the full component $C'_b$ of $S'$ containing $b$ is a subset of $C_b$ and thus $N(C'_b) \cap N(C) \subseteq N(C_b) \cap N(C)$.  
\end{proof}

Our algorithm is the same as Takata's algorithm, expect that we do not output minimal separators with size $>k$, and instead of determining if $\MS(G, C, X)$ is empty we determine if it contains minimal separators of size at most $k$.
For this we use important separators.

\subsection{Important Separators}
We overview the results on important separators~\cite{DBLP:journals/algorithmica/ChenLL09,DBLP:journals/tcs/Marx06} that we use.
This overview is based on the exposition of this technique in~\cite{DBLP:books/sp/CyganFKLMPPS15}.

\begin{definition}
Let $A,B \subseteq V(G)$ be vertex sets of $G$.
A set $S$ is a minimal $A,B$-separator if there exist components $C_A,C_B \in \CC(G \setminus S)$ with $A \subseteq C_A$, $B \subseteq C_B$, and $S = N(C_A) = N(C_B)$.
A minimal $A,B$-separator is an important $A,B$-separator if there is no minimal $A,B$-separator $S'$ such that $|S'| \le |S|$ and $C_A \subsetneq C'_A$, where $C'_A$ is the component of $G \setminus S'$ containing $A$.
\end{definition}

Important separators are exploited by using an enumeration algorithm that given vertex sets $A,B$ and an integer $k$ enumerates important $A,B$-separators of size at most $k$.
In particular, we use the following proposition.

\begin{proposition}[\cite{DBLP:books/sp/CyganFKLMPPS15}]
\label{pro:important}
Let $A,B \subseteq V(G)$ be vertex sets of $G$ and $k$ an integer.
There are at most $4^k$ important $A,B$-separators of $G$ of size at most $k$ and they can be enumerated with polynomial delay.
\end{proposition}

Proposition~\ref{pro:important} will be our tool to check if a subtree of the search tree in the enumeration algorithm is empty.

\subsection{Proof of Theorem~\ref{the:alg}}
Now we are ready to give our algorithm.
We modify Definition~\ref{def:takata} for the purpose of our algorithm.

\begin{definition}
\label{def:mss}
Let $C, X$ be vertex sets $C \subseteq V(G)$ and $X \subseteq N(C)$ so that $a \in C$ and $G[C]$ is connected.
We denote by $\MS(G, k, C, X)$ the set of minimal $a,b$-separators $S$ of $G$ such that $S \subseteq V(G) \setminus C$, $X \subseteq S$, and $|S| \le k$.
\end{definition}

This definition is analogous to Definition~\ref{def:takata}, except that it also includes a size bound $k$.
By the same arguments as given for Takata's algorithm, we can enumerate minimal $a,b$-separators of size at most $k$ with $f(G, k, R)$-delay if we have an $f(G, k, R)$ time algorithm for checking if $\MS(G, k, C, X)$ is empty.
The following lemma shows that we can use important separators to obtain this algorithm.

\begin{lemma}
\label{lem:important}
The set $\MS(G, k, C, X)$ is not empty if and only if there is an important $\{b\},C$-separator $S$ such that $X \subseteq S$ and $|S| \le k$.
\end{lemma}
\begin{proof}
For the if direction we observe that such $S$ satisfies $S \in \MS(G, k, C, X)$.
For the only if direction, let $S' \in \MS(G, k, C, X)$ and denote by $C'_b$ the component of $G \setminus S'$ containing $b$.
If $S'$ is not an important $\{b\},C$-separator then there is an important $\{b\},C$-separator $S$ with $|S| \le k$ and a component $C_b \in \CC(G \setminus S)$ with $C'_b \subseteq C_b$.
Because neither $C_b$ nor $C'_b$ intersects $N(C)$ we have that $N(C'_b) \cap N(C) \subseteq N(C_b) \cap N(C)$, and therefore $X \subseteq S$.
\end{proof}

Lemma~\ref{lem:important} asserts that we can check if $\MS(G, k, C, X)$ is empty by enumerating important $\{b\},C$-separators of size at most $k$.
By Proposition~\ref{pro:important} this can be done in $O^*(4^k)$ time.
To make the time complexity into $O^*(\min(4^k, R))$, where $R$ is the number of minimal separators already outputted, we note that the algorithm of Proposition~\ref{pro:important} works in polynomial delay and all important separators outputted by it are also minimal $a,b$-separators of size at most $k$.
Therefore we simply keep a set of already outputted minimal separators, and if an important separator given by Proposition~\ref{pro:important} is not in this set we output it.
Note that this will cause us to ``miss'' some outputs later, but this does not matter because the outputting of them can be seen just as moved forward.
This completes the proof of Theorem~\ref{the:alg}.

\section{Hardness}
We show that the problem of testing if $\MS(G, k, C, X)$ is empty is NP-hard even in graphs with bipartition $\{\{a\} \cup N(b), \{b\} \cup N(a)\}$, i.e., we prove Theorem~\ref{the:np}.
We reduce from set cover, which is NP-hard~\cite{DBLP:books/fm/GareyJ79}.

Let $U$ be a set and $F$ a family of subsets of $U$.
Given $U$, $F$, and an integer $k$, the set cover problem is to determine if there is a subset $F' \subseteq F$ with $|F'| \le k$ and $U = \bigcup_{T \in F'} T$.

We construct a graph $G(U, F)$ that has four layers, $\{a\}$, $N(a)$, $N(b)$, and $\{b\}$.
The vertices of $N(b)$ corresponds to sets in $F$, i.e., for each set $T \in F$ there is a vertex $v_T \in N(b)$.
For each vertex $v_T \in N(b)$ there are two vertices $u_T,w_T \in N(a)$ that are connected only to $v_T$ and $a$.
The other vertices in $N(a)$ are the elements of $U$.
We add an edge from $z \in U$ to $v_T \in N(b)$ if $z \in T$.

We first show that given a solution to the set cover problem we can construct a minimal $a,b$-separator in $\MS(G(U, F), |U|+|F|+k, \{a\}, U)$.
\begin{lemma}
If there is a subset $F' \subseteq F$ with $|F'| \le k$ and $U = \bigcup_{T \in F'} T$ then $\MS(G(U, F), |U|+|F|+k, \{a\}, U)$ is not empty.
\end{lemma}
\begin{proof}
We construct a minimal $a,b$-separator that consists of vertices $z \in U$, vertices $v_T \in N(b)$ with $T \notin F'$, and vertices $u_T,w_T \in N(a)$ with $T \in F'$.
The size of this separator is $|U| + |F| - k + 2k$.
For any path $a,u_T,v_T,b$ or $a,w_T,v_T,b$ exactly one of the vertices is in the separator, so this indeed separates $a$ from $b$.
This separator is minimal because each $z \in U$ is connected to a vertex $v_T \in N(b)$ with $z \in T \in F'$ implying that $v_T$ is not in the separator.
\end{proof}

We complete the NP-completeness proof by showing that for a minimal $a,b$-separator in $\MS(G(U, F), |U|+|F|+k, \{a\}, U)$ we can construct a solution to the set cover problem.

\begin{lemma}
If there is $S \in \MS(G(U, F), |U|+|F|+k, \{a\}, U)$ then there is a subset $F' \subseteq F$ with $|F'| \le k$ and $U = \bigcup_{T \in F'} T$.
\end{lemma}
\begin{proof}
We construct the subset $F'$ by including the sets $T$ with $v_T \notin S$.
This is a set cover because for each $z \in U$ there must be a vertex $v_T \notin S$ with $z \in T$ because otherwise $S \setminus \{z\}$ would be an $a,b$-separator.
We note that the vertices $S \cap N(b)$ determine $S$ uniquely, in particular forcing $u_T,w_T$ to $S$ if and only if $v_T \notin S$, so we can compute that $|U| + 2|F'| + |F| - |F'| = |S|$.
\end{proof}

This completes the proof of Theorem~\ref{the:np}.

\section{Conclusion}
We gave a fixed-parameter-delay and incremental-polynomial enumeration algorithm for enumerating minimal $a,b$-separators of size at most $k$.
To the best of our knowledge, this is the first algorithm for this problem with non-trivial time complexity.
While our algorithm seems not completely impractical, ideally we would prefer a polynomial-delay algorithm to optimally implement the enumeration as a subroutine in various applications.
We gave an NP-completeness proof that illustrates why our approach falls short in obtaining a polynomial-delay algorithm.
Informally, our proof shows that any algorithm based on Takata's recurrence must have a more ``global'' view on the search space than just the current subtree.
The NP-completeness result also illustrates that enumerating minimal vertex covers of size at most $k$ in a bipartite graph is an important open special case of this problem.

\bibliographystyle{plain}
\bibliography{paper}

\end{document}